\documentclass[11pt]{article}
\usepackage[letterpaper,margin=1in]{geometry}

\usepackage{amsmath,amssymb,amsthm}
\usepackage{microtype}
\usepackage[hidelinks]{hyperref}

\numberwithin{equation}{section}

\newtheorem{theorem}{Theorem}
\newtheorem{lemma}[theorem]{Lemma}
\newtheorem{proposition}[theorem]{Proposition}

\newcommand{\F}{\mathbb F}
\newcommand{\E}{\mathbb E}
\newcommand{\enc}{\operatorname{enc}}
\DeclareMathOperator{\rank}{rank}

\title{Superlogarithmic-Rank Matrix Rigidity for the \\Walsh--Hadamard Transform}
\author{Josh Alman\thanks{Columbia University. \url{josh@cs.columbia.edu}. Supported in part by NSF Grant CCF-2238221 and a Packard Foundation Fellowship.}}
\date{}

\begin{document}
\maketitle

\begin{abstract}
For sufficiently large $N$ which is a power of 2, we prove that changing at
most one percent of the entries of the $N\times N$ Walsh--Hadamard Transform cannot reduce its rank over $\F_3$ to
$\lfloor \log^2 N/80\rfloor$ or below.  To the best of our knowledge,
this is the first constant-fraction rigidity lower bound for an explicit matrix family at a
superlogarithmic target rank over any choice of field, inching toward the parameters in Razborov's program for
communication complexity lower bounds.
\end{abstract}

\section{Introduction}\label{sec:introduction}

For a field
$F$, a matrix $A\in F^{N\times N}$, and a nonnegative integer $r$, the
rank-$r$ rigidity of $A$ over $F$ is the minimum number of entries one
must change to make its rank at most $r$:
$$
 R_A^F(r)=
 \min_{\substack{L\in F^{N\times N}\\ \rank_F L\le r}}
 |\{(i,j)\in[N]^2:A[i,j]\ne L[i,j]\}|.
$$
Matrix rigidity was introduced by Valiant~\cite{Valiant1977} as a route to
circuit lower bounds, and has since found connections to communication
complexity, data structures, coding theory, and algorithms.  For more background,
see Lokam's book~\cite{Lokam2009}, Ramya's survey~\cite{Ramya2020}, and the
course notes of Chou and Golovnev~\cite{ChouGolovnev2020}.  Random matrices are known to be rigid with high probability~\cite{Valiant1977}, but in these applications, one seeks \emph{explicit} matrix families\footnote{A matrix family is explicit if, given $N$, a uniform
algorithm outputs the $N \times N$ matrix $A_N$ in time polynomial in $N$; see~\cite[Open Question 2.1]{Lokam2009} and the discussion following it.} which are rigid.  The quantitative question we study in this paper is how large $r=r(N)$ can be while an
explicit family still satisfies
$$
 R_{A_N}^F(r)\ge \delta N^2
$$
for some fixed $\delta>0$; we call this \emph{constant-fraction rigidity at rank $r$}.

Our main motivation is Razborov's connection between rigidity and
communication complexity~\cite{Razborov1989,Wunderlich2012}.  Roughly, suppose
an $N\times N$ Boolean communication matrix has constant-fraction rigidity at superpolylogarithmic rank
$2^{(\log\log N)^{\omega(1)}}$ over a fixed finite field.  Razborov's result then places the
corresponding function outside $\mathrm{PH}^{cc}$, the
communication-complexity analogue of the polynomial hierarchy.  We call such
a family \emph{Razborov rigid}.

No explicit Razborov rigid family is known.  Prior to this work, for explicit families of matrices, the largest rank at which a constant-fraction rigidity lower bound was
known was only $\Theta(\log N)$.  This is known for code-based and Cauchy
matrices~\cite{Friedman1993,PudlakRodl1994,ShokrollahiSpielmanStemann1997}. Alman and Liang~\cite{AlmanLiang2025} recently proved a nearly optimal bound at this rank for
the Walsh--Hadamard Transform, which is the communication matrix of Inner
Product mod~2, a canonical target for lower bounds against
$\mathrm{PH}^{cc}$. Let
$G=\F_2^n$ and $N=|G|=2^n$.  We index the rows and columns of $H_n$ by
$G$ and set
$$
 x\cdot y=\sum_{i=1}^n x_i y_i\in\F_2,
 \qquad
 H_n[x,y]=(-1)^{x\cdot y}.
$$
We regard $1$ and $-1$ as elements of $\F_3$, so
$H_n\in\F_3^{G\times G}$.  Alman and Liang proved that
\begin{equation}\label{eq:ALbound}
 R_{H_n}^{\F_3}(c\log N)\ge
 \left(\frac12-o(1)\right)N^2
\end{equation}
for some constant
$c>0$. The leading constant $1/2$ is optimal, since replacing $H_n$ by
the all-ones matrix gives rank $1$ after changing
$(1/2-o(1))N^2$ entries.

In this paper, we prove a new matrix rigidity lower bound for $H_n$:
\begin{theorem}\label{thm:main}
For every integer $n\ge2000$, letting $N=2^n$,
$$
 R_{H_n}^{\F_3}\!\left(
 \left\lfloor\frac{\log^2 N}{80}\right\rfloor
 \right)
 >\frac{N^2}{100}.
$$
\end{theorem}
To the best of our knowledge, this is the first constant-fraction rigidity
lower bound for an explicit matrix family at a superlogarithmic target rank over any field. The aforementioned code-based and Cauchy matrix bounds are only
$R_{A_N}(r) \geq \Omega(N^2 \log(N/r) / r)$ for
$r \geq \Omega(\log N)$, which is $o(N^2)$ when
$r=\omega(\log N)$.

This also makes progress on a hardness-amplification target of Alman and
Liang.  They showed, roughly, that improving the lower bound of Equation~(\ref{eq:ALbound}) to rank $(\log N)^{1+\varepsilon}$ would amplify to Razborov rigidity (see~\cite[Theorem 1.7]{AlmanLiang2025} for the precise statement).  Theorem~\ref{thm:main} gives the first constant-fraction rigidity at the needed rank, but proves a $1/100$ error
fraction rather than near-$1/2$.

The low-rank rigidity of $H_n$ also has recent algorithmic significance.  Alman and Rao~\cite{AlmanRao2023} used a constant-rank rigidity upper bound for the Walsh--Hadamard transform to give the smallest-known arithmetic circuits for both $H_n$ and the Discrete Fourier Transform on power-of-2-sized inputs.  Better rigidity upper bounds could lead to faster algorithms along this
route, while lower bounds show limits on this approach; see also \cite[Section~6.2]{AlmanLiang2025} for more details.

Alman and Rao's circuit constructions built on the work of Alman and
Williams~\cite{AlmanWilliams2017} who showed that the Walsh--Hadamard Transform is not \emph{Valiant rigid}: for every sufficiently small
constant $\varepsilon>0$ and every field $F$, there is a constant
$c_\varepsilon>0$ such that
$$
 R_{H_n}^{F}\!\left(\left\lfloor N^{1-c_\varepsilon}\right\rfloor\right)
 \le N^{1+\varepsilon}.
$$
This concerns a much larger target rank and does not rule out Razborov rigidity.

\paragraph*{Notation}
Throughout, all logarithms are base two.  For a finite set $S$ and field
$F$, write $F^S$ for the set of vectors of length $|S|$, and for $v\in F^S$ and $x\in S$, write
$v_x$ for the coordinate at $x$.  For $u,v\in\mathbb C^S$, use
$\langle u,v\rangle=\sum_{x\in S}\overline{u_x}v_x$, where $\overline{u_x}$ is the complex conjugate of $u_x$, and
$\|u\|_2=\sqrt{\langle u,u\rangle}$.

\section{The organizing correlation inequality}\label{sec:few-outputs}

The proof is organized around the following inequality.

\begin{lemma}[Small output set versus small average overlap]
\label{lem:few-outputs}
Let $\Omega$ be a nonempty finite set and $d$ be a positive integer. For every $\gamma\in\Omega$, let
$p_\gamma,q_\gamma\in\mathbb C^d$ be unit vectors.  Suppose that the
vectors $q_\gamma$ take at most $M$ distinct values, where $M$ is a
positive integer.  Then, for every positive integer $s$,
\begin{equation}
 \left|\E_{\alpha\in\Omega}\langle q_\alpha,p_\alpha\rangle\right|^{2s}
 \le
 M\max_{\gamma\in\Omega}
 \E_{\beta\in\Omega}|\langle p_\gamma,p_\beta\rangle|^s,
 \label{eq:few-outputs}
\end{equation}
where the expectations are over $\alpha,\beta$ drawn independently and uniformly from $\Omega$.
\end{lemma}

The inequality says that large correlation between the matched pairs
$(q_\alpha,p_\alpha)$ requires either many possible $q_\alpha$'s or
large average powered overlap among the $p_\alpha$'s.

In Subsection~\ref{subsec:probes}, we will construct the vectors
$p_\alpha,q_\alpha$ from random signed column sums of $H_n$ and of a
hypothetical low-rank approximation.  Low rank will make the set of possible
$q_\alpha$'s small, bounding $M$.  In Subsection~\ref{subsec:signal}, the
sparsity of the approximation error will give a lower bound on the left side of
\eqref{eq:few-outputs}.  In Subsection~\ref{subsec:overlap}, we will use the structure of $H_n$ to analyze the vectors $p_\alpha$ by expanding them in the normalized
columns of $H_n$.  The orthogonality and coordinatewise multiplication rule of
these columns, together with a tail bound for intersections of random subspaces,
will give an upper bound on the right side of
\eqref{eq:few-outputs}.  We combine the bounds in Subsection~\ref{subsec:combine}.

Lemma~\ref{lem:few-outputs} follows from a generalized Bessel inequality, also known as the
almost-orthogonality or $TT^*$ method; see
Tao's exposition~\cite[Proposition~2 and Remark~3]{Tao2015LargeSieve}. Nonetheless, for completeness we directly prove Lemma~\ref{lem:few-outputs} in Section~\ref{sec:few-outputs-proof} below.

\section{Proof of Theorem~\ref{thm:main}}\label{sec:main-proof}

\begin{proof}
Assume to the contrary that $L\in\F_3^{G\times G}$ satisfies
$$
 r:=\rank_{\F_3}L
 \le\left\lfloor\frac{n^2}{80}\right\rfloor
$$
and differs from $H_n$ in at most $N^2/100$ positions.  Put
$$
 \Delta:=H_n-L\in\F_3^{G\times G}.
$$
All matrix-vector products below are over $\F_3$.

\subsection{Random signed column probes}\label{subsec:probes}

Choose
\begin{equation}
 k:=2\left\lfloor\frac n8\right\rfloor.
 \label{eq:k}
\end{equation}
Thus $k$ is even and $k\le n/4$.

Let $\Omega$ be the set of pairs
$$
 \alpha=(y_1,\ldots,y_k;\varepsilon_1,\ldots,\varepsilon_k),
$$
where $(y_1,\ldots,y_k)\in G^k$ is an ordered tuple of linearly independent vectors,
and each $\varepsilon_j\in\{1,-1\}$.  We draw $\alpha$ uniformly from
$\Omega$ in the probabilistic arguments below.

For $y\in G$, let $\mathbf e_y\in\F_3^G$ be the coordinate indicator
vector of $y$: it has a $1$ in coordinate $y$ and $0$ elsewhere.
For $\alpha\in\Omega$, define the probe vector
\begin{equation}
 a_\alpha:=\sum_{j=1}^k\varepsilon_j\mathbf e_{y_j}\in\F_3^G.
 \label{eq:signed-column-sum}
\end{equation}
This is a $k$-sparse vector whose nonzero entries have random $\{-1,1\}$ values. For every $A\in\F_3^{G\times G}$,
$$
 Aa_\alpha=\sum_{j=1}^k\varepsilon_j A[\mathord\cdot,y_j]\in\F_3^G,
$$
where $A[\mathord\cdot,y_j]$ is the column indexed by $y_j$.

Let $\omega:=e^{2\pi\mathrm i/3}\in\mathbb C$ be the cube root of unity, where $\mathrm i^2=-1$.
Interpreting elements of $\F_3$ as exponents modulo $3$, define the phase
encoding $\enc:\F_3^G\to\mathbb C^G$ by
$$
 \enc(v)_x:=N^{-1/2}\omega^{v_x}\qquad(x\in G).
$$
For each $\alpha\in\Omega$, define the unit vectors
$$
 p_\alpha:=\enc(H_n a_\alpha),
 \qquad
 q_\alpha:=\enc(La_\alpha)
 \quad\text{in }\mathbb C^G.
$$
The image of $L$ is an $r$-dimensional space over $\F_3$, so it has
$3^r$ elements.  Since every $La_\alpha$ lies in this space, the vectors
$q_\alpha$ take at most $3^r$ values.  Thus low rank controls the factor
$M$ in Lemma~\ref{lem:few-outputs}.

\subsection{A correlation signal survives arbitrary errors}\label{subsec:signal}

For every $a\in\F_3^G$, define the complex correlation between the unit vectors
\begin{align}
 z(a)
 &:=\langle\enc(La),\enc(H_n a)\rangle \notag\\
 &=\frac1N\sum_{x\in G}
   \overline{\omega^{(La)_x}}\,\omega^{(H_n a)_x} \notag\\
 &=\frac1N\sum_{x\in G}\omega^{(\Delta a)_x}.
 \label{eq:z-formula}
\end{align}
For $\alpha\in\Omega$, set
$$
 z_\alpha:=z(a_\alpha)=\langle q_\alpha,p_\alpha\rangle.
$$

\begin{lemma}[Correlation lower bound]\label{lem:signal}
\begin{equation}
 \operatorname{Re}\E_{\alpha\in\Omega}z_\alpha
 >
 \frac12\left(\frac{19}{20}\right)^k.
 \label{eq:signal}
\end{equation}
\end{lemma}

\begin{proof}
Let
$$
 \delta:=\frac{|\{(x,y)\in G^2:\Delta[x,y]\ne0\}|}{N^2}
 \le\frac1{100}
$$
be the fraction of entries where $H_n$ and $L$ differ.  Before arguing about $\Omega$, we first consider the variant where we sample columns
with replacement which need not be linearly independent. Choose $Y_1,\ldots,Y_k$ independently and uniformly from
$G$, together with independent uniform signs
$\varepsilon_j\in\{1,-1\}$, and put
$$
 a=\sum_{j=1}^k\varepsilon_j\mathbf e_{Y_j}\in\F_3^G.
$$
For a fixed row $x\in G$, let $d_x\in[0,1]$ be the fraction of nonzero
entries in that row of $\Delta$.  For each $j$, consider the random phase
$\omega^{\varepsilon_j\Delta[x,Y_j]}$.  If
$\Delta[x,Y_j]=0$, this phase is $1$.  If
$\Delta[x,Y_j]\ne0$, then this entry is $1$ or $-1$ in $\F_3$, and
averaging over the random sign $\varepsilon_j$ gives that the expected value
of the phase is
$$
 \frac{\omega+\omega^{-1}}2=-\frac12.
$$
Since a fraction $1-d_x$ of the row entries are zero and a fraction $d_x$
are nonzero,
$$
 \E_{Y_j,\varepsilon_j}
 \omega^{\varepsilon_j\Delta[x,Y_j]}
 =(1-d_x)\cdot1+d_x\cdot\left(-\frac12\right)
 =1-\frac32d_x.
$$
Thus the random signs make the two possible nonzero errors indistinguishable:
the expectation depends only on whether the sampled entry was changed.  Since
$$
 \omega^{(\Delta a)_x}
 =\prod_{j=1}^k\omega^{\varepsilon_j\Delta[x,Y_j]},
$$
independence of the $k$ pairs $(Y_j,\varepsilon_j)$ and
\eqref{eq:z-formula} give
$$
 \E z(a)=\frac1N\sum_{x\in G}\left(1-\frac32d_x\right)^k.
$$
The average of the $d_x$'s is $\delta$.  Some values
$1-\tfrac32d_x$ may be negative, but since $k$ is even,
$\rho\mapsto\rho^k$ is convex on all of $\mathbb R$, and so Jensen's inequality
gives
\begin{equation}
 \E z(a)
 \ge\left(1-\frac32\delta\right)^k
 \ge\left(\frac{197}{200}\right)^k.
 \label{eq:independent-signal}
\end{equation}

This is nearly our desired result, but we need to modify it to apply to $\Omega$.
Let $I$ be the event that $Y_1,\ldots,Y_k$ are linearly independent, and
write $\eta:=\Pr(I^c)$.  We bound $\eta$ by a standard argument: If the first $j$ choices are independent, the next
choice lies in their span with probability $2^{j-n}$, so a union bound gives
$$
 \eta\le\sum_{j=0}^{k-1}2^{j-n}<2^{k-n}.
$$
Since $k\le n/4$ and $n\ge2000$, we have $n>2k+1$, and hence
\begin{equation}
 \eta<2^{k-n}<2^{-k-1}
 \le\frac12\left(\frac{19}{20}\right)^k.
 \label{eq:eta-small}
\end{equation}
Because $z(a)$ is an average of complex numbers of magnitude $1$, we know
$\operatorname{Re}\E[z(a)\mid I^c]\le1$.  Taking real parts in
$$
 \E z(a)=(1-\eta)\E[z(a)\mid I]+\eta\E[z(a)\mid I^c]
$$
and using \eqref{eq:independent-signal} gives
$$
 (1-\eta)\operatorname{Re}\E[z(a)\mid I]
 \ge\left(\frac{197}{200}\right)^k-\eta.
$$
By \eqref{eq:eta-small}, the right side is positive, and
\begin{align*}
 \operatorname{Re}\E[z(a)\mid I]
 &\ge \frac{(197/200)^k-\eta}{1-\eta}\\
 &\ge \left(\frac{197}{200}\right)^k-\eta\\
 &>\left(\frac{197}{200}\right)^k
   -\frac12\left(\frac{19}{20}\right)^k\\
 &>\frac12\left(\frac{19}{20}\right)^k,
\end{align*}
where the last inequality uses $197/200>19/20$.
Conditioned on $I$, the tuple $(Y_1,\ldots,Y_k)$ is uniform among all ordered
linearly independent $k$-tuples, while the signs remain independent and
uniform.  The conditioned pair therefore has exactly the uniform distribution
on $\Omega$, proving the lemma.
\end{proof}

\subsection{Column expansions and overlap}\label{subsec:overlap}

For
$\alpha=(y_1,\ldots,y_k;\varepsilon_1,\ldots,\varepsilon_k)\in\Omega$,
define
$$
 U_\alpha:=\operatorname{span}_{\F_2}\{y_1,\ldots,y_k\}\le G.
$$
This is a $k$-dimensional subspace (since $y_1,\ldots,y_k$ are linearly independent).  Every such subspace has the
same number of ordered bases, so $U_\alpha$ is uniform among the
$k$-dimensional subspaces of $G$.

For $y\in G$, let $h_y\in\mathbb C^G$ be the column of $H_n$
indexed by $y$, viewed as a complex vector:
$$
 (h_y)_x=H_n[x,y]=(-1)^{x\cdot y}\qquad(x\in G).
$$
The normalized columns $N^{-1/2}h_y$, $y\in G$, form an orthonormal
basis of $\mathbb C^G$: for $y,y'\in G$,
$$
 \left\langle N^{-1/2}h_y,N^{-1/2}h_{y'}\right\rangle
 =\frac1N\sum_{x\in G}(-1)^{x\cdot(y-y')}
 =\begin{cases}
 1,&y=y',\\
 0,&y\ne y'.
 \end{cases}
$$
Indeed, when $y-y'\ne0$, the value $x\cdot(y-y')$ equals $1$ for
exactly half the choices of $x\in G$.

Since the normalized columns form a basis, there are unique coefficients
$b_{\alpha,y}\in\mathbb C$ such that
$$
 p_\alpha=\sum_{y\in G}b_{\alpha,y}N^{-1/2}h_y.
$$
For a subspace $W\le G$, define
$$
 p_{\alpha,W}:=\sum_{y\in W}b_{\alpha,y}N^{-1/2}h_y.
$$
In other words, $p_{\alpha,W}$ is obtained from the expansion of $p_\alpha$
in the (normalized) columns of $H_n$ by keeping only the columns whose indices
lie in $W$.

\begin{lemma}[Weight on columns indexed by a subspace]
\label{lem:column-weight}
If $\alpha\in\Omega$ and $W\le U_\alpha$ is a subspace with dimension $t$, then
\begin{equation}
 \|p_{\alpha,W}\|_2^2
 \le\left(\frac34\right)^{k-t}.
 \label{eq:column-weight}
\end{equation}
\end{lemma}

\begin{proof}
Write $\alpha=(y_1,\ldots,y_k;\varepsilon_1,\ldots,\varepsilon_k)$.  For
$u=(u_1,\ldots,u_k)\in\F_2^k$, define
$$
 \phi_\alpha(u):=\sum_{j=1}^k u_jy_j
 =\sum_{j:u_j=1}y_j\in U_\alpha,
$$
and write $|u|$ to denote the number of coordinates of $u$ equal to $1$.  Since
$y_1,\ldots,y_k$ form a basis of $U_\alpha$, the map
$\phi_\alpha:\F_2^k\to U_\alpha$ is a linear bijection.

For each $x\in G$,
$$
 (H_n a_\alpha)_x=\sum_{j=1}^k\varepsilon_j(h_{y_j})_x
 \quad\text{in }\F_3,
 \qquad
 (p_\alpha)_x=N^{-1/2}\prod_{j=1}^k
 \omega^{\varepsilon_j(h_{y_j})_x}.
$$
For $\sigma\in\{1,-1\}$, we can verify that $\omega=-\frac12+\frac{\mathrm i\sqrt3}{2}$ satisfies
$$
 \omega^\sigma=-\frac12+\frac{\mathrm i\sqrt3}{2}\,\sigma.
$$
Substituting this identity gives
$$
 (p_\alpha)_x=N^{-1/2}\prod_{j=1}^k
 \left(-\frac12+\frac{\mathrm i\sqrt3}{2}\,
 \varepsilon_j(h_{y_j})_x\right).
$$
When this product is expanded, $u_j=0$ means that we choose the constant
term $-1/2$ from the $j$th factor, while $u_j=1$ means that we choose
$(\mathrm i\sqrt3/2)\varepsilon_j(h_{y_j})_x$.  Using, for every
$u\in\F_2^k$ and $x\in G$,
$$
 \prod_{j:u_j=1}(h_{y_j})_x
 =(-1)^{x\cdot\sum_{j:u_j=1}y_j}
 =(h_{\phi_\alpha(u)})_x,
$$
we obtain
\begin{equation}
 p_\alpha
 =
 \sum_{u\in\F_2^k}
 \left(-\frac12\right)^{k-|u|}
 \left(\frac{\mathrm i\sqrt3}{2}\right)^{|u|}
 \left(\prod_{j:u_j=1}\varepsilon_j\right)
 N^{-1/2}h_{\phi_\alpha(u)}.
 \label{eq:column-expansion}
\end{equation}
As $u$ ranges over $\F_2^k$, the bijection $\phi_\alpha$ indexes each element
of $U_\alpha$ exactly once.  The coefficient of the normalized column indexed
by $\phi_\alpha(u)$ in (\ref{eq:column-expansion}) has squared magnitude
\begin{equation}
 \left(\frac14\right)^{k-|u|}
 \left(\frac34\right)^{|u|}.
 \label{eq:coefficient-weight}
\end{equation}
These values are exactly the probability mass function of a random vector
$X=(X_1,\ldots,X_k)\in\F_2^k$ with independent coordinates satisfying
$\Pr(X_j=1)=3/4$ for all $j$.

Since $\phi_\alpha$ is a linear bijection and $\dim W=t$, the subspace
$D:=\phi_\alpha^{-1}(W)\le\F_2^k$ also has dimension~$t$.
Equations~\eqref{eq:column-expansion} and \eqref{eq:coefficient-weight} give
$$
 \|p_{\alpha,W}\|_2^2=\Pr(X\in D).
$$
By Gaussian elimination, there is a set
$J\subseteq\{1,\ldots,k\}$ of $t$ coordinates for which the projection
$D\to\F_2^J$ is bijective.  Fix $x_J\in\F_2^J$.  Conditioned on
$X_J=x_J$, membership in $D$ forces the remaining $k-t$ coordinates of $X$
to equal one particular pattern.  Those coordinates remain independent, and
each matches its required value with probability at most $3/4$.  Hence
$$
 \Pr(X\in D\mid X_J=x_J)
 \le\left(\frac34\right)^{k-t}.
$$
Averaging over $X_J$ proves the lemma.
\end{proof}

For $\alpha,\beta\in\Omega$, set
$$
 W:=U_\alpha\cap U_\beta\le G,
 \qquad
 t:=\dim_{\F_2}W\in\{0,\ldots,k\}.
$$
Equation~\eqref{eq:column-expansion} shows that the expansion of
$p_\alpha$ uses only normalized columns indexed by $U_\alpha$, and the
same holds for $p_\beta$ with $U_\beta$.  The only columns that appear in
both expansions are therefore those indexed by $W$.  Since the normalized
columns are orthonormal, Lemma~\ref{lem:column-weight} and Cauchy--Schwarz give
\begin{equation}
 \begin{aligned}
 |\langle p_\alpha,p_\beta\rangle|
 &=|\langle p_{\alpha,W},p_{\beta,W}\rangle|\\
 &\le\|p_{\alpha,W}\|_2\,\|p_{\beta,W}\|_2\\
 &\le\left(\frac34\right)^{k-t}.
 \end{aligned}
 \label{eq:overlap}
\end{equation}
Thus a large overlap is possible only when the two subspaces indexing
these column expansions have a large intersection.

We next bound the intersection dimension of a random subspace with a fixed
one.

\begin{lemma}[Tail bound for random subspace intersections]
\label{lem:intersection}
Fix a $k$-dimensional subspace $U\le G$, and let $V$ be a uniformly random
$k$-dimensional subspace of $G$.  For every $0\le t\le k$,
\begin{equation}
 \Pr\bigl(\dim(U\cap V)\ge t\bigr)
 \le2^{-(n-2k)t}.
 \label{eq:intersection-tail}
\end{equation}
\end{lemma}

\begin{proof}
The case $t=0$ is immediate, so assume $1\le t\le k$.  If
$\dim(U\cap V)\ge t$, then $V$ contains an ordered linearly independent
$t$-tuple of vectors from $U$.  There are
$$
 \prod_{i=0}^{t-1}(2^k-2^i)\le2^{kt}
$$
such tuples, since after $i$ independent vectors have been chosen, their span
contains $2^i$ vectors.

Fix one of them, $(w_1,\ldots,w_t) \in U^t$.  A double count of pairs consisting of a
$k$-dimensional subspace and an ordered independent $t$-tuple inside it gives
$$
 \Pr(\{w_1,\ldots,w_t\}\subseteq V)
 =\prod_{i=0}^{t-1}\frac{2^k-2^i}{2^n-2^i}.
$$
Indeed, every $k$-dimensional subspace contains
$\prod_{i=0}^{t-1}(2^k-2^i)$ such tuples, while $G$ contains
$\prod_{i=0}^{t-1}(2^n-2^i)$. Moreover, invertible linear maps act
transitively on the ordered independent $t$-tuples in $G$, so each tuple is
contained in the same number of $k$-dimensional subspaces.

For every $0\le i<t$,
$$
 \frac{2^k-2^i}{2^n-2^i}
 =2^{k-n}\frac{1-2^{i-k}}{1-2^{i-n}}
 \le2^{k-n}.
$$
Thus a fixed tuple is contained in $V$ with probability at most
$2^{-(n-k)t}$.  A union bound over the $\leq 2^{kt}$ possible tuples gives as desired that
\begin{displaymath}
 \Pr\bigl(\dim(U\cap V)\ge t\bigr)
 \le2^{kt}2^{-(n-k)t}
 =2^{-(n-2k)t}.\qedhere
\end{displaymath}
\end{proof}

We now prove our high-moment estimate. Set
$$
 g:=n-2k,
 \qquad
 s:=3g.
$$
The preceding lemma says that an intersection of dimension at least $t$
has probability at most $2^{-gt}$.  Since $k\le n/4$, we have
$g\ge n/2\ge1$, so $s$ is a positive integer.

\begin{proposition}[High-moment bound for the overlaps]
\label{prop:average-overlap}
\begin{equation}
 \max_{\alpha\in\Omega}
 \E_{\beta\in\Omega}|\langle p_\alpha,p_\beta\rangle|^s
 <8 \cdot 2^{-gk}.
 \label{eq:average-overlap}
\end{equation}
\end{proposition}

\begin{proof}
Fix any $\alpha$, and in the arguments below, choose $\beta$ uniformly from $\Omega$. Let
$$
 T:=\dim_{\F_2}(U_\alpha\cap U_\beta)\in\{0,\ldots,k\}.
$$
The subspace $U_\beta$ is uniform among the $k$-dimensional subspaces of $G$.
For every $t$, equation~\eqref{eq:overlap} and
Lemma~\ref{lem:intersection} give
$$
 \Pr(T=t)\le\Pr(T\ge t)\le2^{-gt},
 \qquad
 |\langle p_\alpha,p_\beta\rangle|^s
 \le\left(\frac34\right)^{s(k-t)}
 \quad\text{when }T=t.
$$
The choice $s=3g$ balances these bounds:
\begin{align*}
 2^{-gt}\left(\frac34\right)^{s(k-t)}
 &=2^{-gt}\left(\frac{27}{64}\right)^{g(k-t)}\\
 &=2^{-gk}\left(\frac{27}{32}\right)^{g(k-t)}
 \le2^{-gk}.
\end{align*}
Keeping the factor $(27/32)^{g(k-t)}$ from the preceding display,
rather than replacing it by $1$, and writing $\ell=k-t$, we obtain
\begin{align*}
 \E_\beta|\langle p_\alpha,p_\beta\rangle|^s
 &\le\sum_{t=0}^k
 2^{-gt}\left(\frac34\right)^{s(k-t)}\\
 &=2^{-gk}\sum_{\ell=0}^k
 \left(\frac{27}{32}\right)^{g\ell}\\
 &\le2^{-gk}\sum_{\ell=0}^{\infty}
 \left(\frac{27}{32}\right)^\ell
 =\frac{32}{5}\,2^{-gk}
 <8 \cdot 2^{-gk}.
\end{align*}
Here we used $g\ge1$.  Since this held for any arbitrary $\alpha$, this proves the
proposition.
\end{proof}

\subsection{Combining the bounds}\label{subsec:combine}

Lemma~\ref{lem:signal} gives
$$
 \left|\E_\alpha\langle q_\alpha,p_\alpha\rangle\right|
 \ge\operatorname{Re}\E_\alpha z_\alpha
 >\frac12\left(\frac{19}{20}\right)^k.
$$
The vectors $q_\alpha$ take at most $3^r$ values, so by
Lemma~\ref{lem:few-outputs} with $M=3^r$, and then
Proposition~\ref{prop:average-overlap},
\begin{equation}
 \left[\frac12\left(\frac{19}{20}\right)^k\right]^{2s}
 <8\cdot3^r\,2^{-gk}.
 \label{eq:decisive}
\end{equation}
Since $n\ge2000$, equation~\eqref{eq:k} gives
$k\ge n/4-2\ge n/5$.  Together with $g\ge n/2$, this gives
$gk\ge n^2/10$. We also have $g\le n$.  Numerically,
$\log(20/19)\approx0.074<1/8$.  Taking logarithms in
\eqref{eq:decisive} and using $s=3g$ gives
\begin{align*}
 r\log 3
 >gk-3-2s\left(1+k\log\frac{20}{19}\right)
 >\frac14gk-6g-3
 \ge\frac{n^2}{40}-6n-3
 >\frac{n^2}{50}.
\end{align*}
The assumed rank bound $r \le\left\lfloor n^2 / 80\right\rfloor$ and $\log 3<8/5$, however, give
$r\log 3<n^2/50$, a contradiction.  This proves the theorem.
\end{proof}

\paragraph{Comparison with prior rigidity lower bounds for $H_n$.}
Like de Wolf's quantum-inspired proof of Hadamard rigidity~\cite{deWolf2006}, we pass
to complex unit vectors and turn entrywise closeness into inner products.
De Wolf encodes rows as low-dimensional quantum states, whereas here we use low rank
$r$ to give only $3^r$ possible encoded outputs of the random probes.
Alman and Liang~\cite[Lemma~4.1]{AlmanLiang2025} also use roots of unity, but
lift the entire finite-field approximation by polynomial interpolation and
then apply a singular-value bound; we instead expand random signed outputs
in the normalized columns of $H_n$ and bound their tensor-powered overlaps.
Unlike classical
$H_n$ rigidity proofs, our argument does not extract an unchanged
high-rank submatrix~\cite{KashinRazborov1998,Lokam2001,deWolf2006}.
\section{Proof of Lemma~\ref{lem:few-outputs}}
\label{sec:few-outputs-proof}

\begin{proof}
Let
$$
 \mathcal C:=\{q_\alpha:\alpha\in\Omega\},\qquad
 P_\alpha:=p_\alpha^{\otimes s},\qquad
 Q_q:=q^{\otimes s}\quad(q\in\mathcal C).
$$
Then $\langle Q_q,P_\alpha\rangle=\langle q,p_\alpha\rangle^s$.  Set
$b:=|\E_\alpha\langle q_\alpha,p_\alpha\rangle|$.  By the triangle
inequality and Jensen's inequality,
$$
 b^s
 \le\left(\E_\alpha|\langle q_\alpha,p_\alpha\rangle|\right)^s
 \le\E_\alpha|\langle q_\alpha,p_\alpha\rangle|^s
 =\E_\alpha|\langle Q_{q_\alpha},P_\alpha\rangle|.
$$
For each $\alpha$, let $\widetilde P_\alpha$ be the vector gotten by multiplying $P_\alpha$ by the appropriate unit complex scalar so that
$$
 \langle Q_{q_\alpha},\widetilde P_\alpha\rangle
 =|\langle Q_{q_\alpha},P_\alpha\rangle|.
$$
For $q\in\mathcal C$, set
$$
 S_q:=\E_\alpha[\mathbf 1_{\{q_\alpha=q\}}\widetilde P_\alpha].
$$
Then
$$
 \langle Q_q,S_q\rangle
 =\E_\alpha\!\left[\mathbf 1_{\{q_\alpha=q\}}
   |\langle Q_q,P_\alpha\rangle|\right]\ge0,
$$
and Cauchy--Schwarz gives
$$
 b^s\le\sum_{q\in\mathcal C}\langle Q_q,S_q\rangle
 \le\left(|\mathcal C|\sum_{q\in\mathcal C}\|S_q\|_2^2\right)^{1/2}.
$$
For independent uniform $\alpha,\beta\in\Omega$, expanding the squared
norms gives the real nonnegative quantity
\begin{align*}
 \sum_{q\in\mathcal C}\|S_q\|_2^2
 &=\E_{\alpha,\beta}\!
   \left[\mathbf 1_{\{q_\alpha=q_\beta\}}
   \langle\widetilde P_\alpha,\widetilde P_\beta\rangle\right]\\
 &\le\E_{\alpha,\beta}\!
   \left[\mathbf 1_{\{q_\alpha=q_\beta\}}
   |\langle P_\alpha,P_\beta\rangle|\right]\\
 &\le\E_{\alpha,\beta}|\langle p_\alpha,p_\beta\rangle|^s
 \le\max_{\gamma\in\Omega}\E_\beta
   |\langle p_\gamma,p_\beta\rangle|^s.
\end{align*}
Squaring and using $|\mathcal C|\le M$ proves \eqref{eq:few-outputs}.
\end{proof}

\section*{AI Disclosure}

Theorem~\ref{thm:main} was originally proved by the author without AI assistance, although the proof was substantially longer and more complicated than the argument presented here, including additional spectral and probabilistic arguments. The author later asked ChatGPT 5.6 whether it could simplify the proof, and (after some back-and-forth) ChatGPT found a number of simplifying ideas which are present in this paper. Most notably, it proposed Lemma~\ref{lem:few-outputs} using a bound on $M$ as a simpler way to make use of the low-rank part of a rigidity decomposition, and suggested modifications to later lemmas to go with this. ChatGPT also gave a number of writing suggestions and corrections, although the author ultimately takes full responsibility for the paper's contents.

\bibliographystyle{alpha}
\bibliography{walsh_f3_quadratic_rigidity_references}

\end{document}